\documentclass[11pt]{article}
\usepackage{natbib,amssymb,amsmath,fullpage,hyperref,graphicx}
\usepackage{amsthm}
\usepackage{gensymb}
\usepackage{caption,textcomp}
\usepackage{subcaption}
\usepackage{multirow}

\newcommand{\Exp}[1]{{\text{E}}[ \ensuremath{ #1 } ]  }
\newcommand{\Var}[1]{{\text{Var}}[ \ensuremath{ #1 } ]  }

\newcommand{\bs}[1]{{\boldsymbol #1}}

\newtheorem{theorem}{Theorem}[section]

\newtheorem{lemma}[theorem]{Lemma}

\newtheorem{proposition}{Proposition}[section]

\begin{document}

\title{Adaptive Sign Error Control}
\author{Chaoyu Yu$^{1}$ and  Peter D. Hoff$^{2}$ \\
	$^1$Department of Biostatistics, University of Washington-Seattle \\
	$^2$Department of Statistical Science, Duke University}
\maketitle

\begin{abstract}
%	When testing a null hypothesis $H_0: \theta = 0$, the sign of $\theta$ is usually inferred when a rejection is made. Typically, the type I error rate is explicitly controlled, while the probability of a sign error is not. In this paper, we propose instead to select the level $\alpha$ of a test in order to control its sign error rate. For multiple testing scenarios, 
	
	In multiple testing scenarios, typically the sign of a parameter is inferred when its estimate exceeds some significance threshold in absolute value. Typically, the significance threshold is chosen to control the experimentwise type I error rate, family-wise type I error rate or the false discovery rate. However, controlling these error rates does not explicitly control the sign error rate. 
%The lack of sign error rate control makes the sign estimates unreliable. 
In this paper, we propose two procedures for adaptively selecting an experimentwise significance threshold in order to control the sign error rate. The first controls the sign error rate conservatively, without any distributional assumptions on the parameters of interest. The second is an empirical Bayes procedure, and achieves optimal performance asymptotically when a model for the 
distribution of the parameters is correctly specified. We also discuss an adaptive procedure to minimize the sign error rate when the experimentwise type I error rate is held fixed.

\end{abstract}

\section{Introduction}
We consider multiparameter inference for the normal means model,
\begin{equation}
	\boldsymbol{Y} | \boldsymbol{\theta} \sim N(\boldsymbol{\theta}, \boldsymbol{I}),
	\label{m1}
\end{equation}
where $\boldsymbol{Y} = (Y_1,\ldots,Y_m)$ and $\boldsymbol{\theta} = (\theta_1,\ldots, \theta_m)$. Simultaneous inference for $\theta_1,\ldots, \theta_m$ often begins by testing $H_i:\theta_i=0$ 
for each $i =1,\ldots, m$ at level $\alpha$, that is, we reject 
$H_i$ if $|Y_i|$ exceeds the $1-\alpha/2$ standard normal quantile,  $z_{1-\alpha/2}$. This controls 
the experimentwise type I error rate to be equal to $\alpha$. A popular method for choosing $\alpha$ is the Benjamini Hochberg (BH) procedure \citep{benjamini1995controlling}. 
The BH procedure is an adaptive method for selecting a value of 
$\alpha$ that will bound the false discovery rate (FDR), which is defined as $
\text{FDR} = E[\tfrac{V}{R \vee 1}|\theta_1,...,\theta_p]$,
where $R$ is the number of rejections and $V$ is the number of false 
rejections, that is, the number of null hypotheses that are rejected but true. 
There is a large literature on FDR control, see \citet{efron2012large}, \citet{benjamini2010discovering}, \citet{genovese2004stochastic}, \citet{storey2002direct} and \citet{storey2007optimal}. 
%False discovery rate control procedures bound the expected proportion of rejections or ``discoveries'' that are actually nulls. 
However, in many applications it is likely that none of the $\theta_i$'s are truly equal to exactly zero. For example, 
in the case where each $Y_i$ represents a difference in sample averages 
between two treatments, \citet{tukey1991} argued that evaluating if 
$\theta_i=0$ is ``foolish''
since the effects of two different factors are always different, however minutely.  In such cases, \citet{tukey1962} suggests that a more meaningful task is to judge whether or not  there is enough evidence to infer the sign of $\theta_i$, instead of whether or not it is zero. However, if significance tests 
are used in this way, then FDR control is inappropriate since it 
is always zero if there are no true nulls. 
Instead, the relevant error control is not the FDR, but
a sign error rate \citep{gelman2000type, gelman2014beyond, owen2016confidence}.

\citet{Benjamini2005} showed that the Benjamini-Hochberg algorithm can be used to control the pure directional FDR, defined as the expected proportion of discoveries in which a positive parameter is declared negative or a negative parameter is declared positive. We refer to this procedure as the BY procedure in this paper. Some follow-up work includes \citet{zhao2015} who used weighted $p$-value methods, and \citet{guo2010} who extended the idea to making multidimensional directional decisions. \citet{weinstein2013selection} derived new selection-adjusted confidence intervals by minimizing an objective function comprised of the length of the acceptance region and a penalty term for the magnitude of the observation. They showed in examples that these procedures have correct coverage on selected parameters, and have more power to determine the sign, but they did not assess the sign error rate directly. These procedures also do not utilize information across experiments and so are not adaptive.
% and tend to be conservative in terms of the number of signs inferred. 
\citet{Stephens2016} proposed an empirical Bayes procedure for sign error control to gain more power. However, the focus there was control of the local sign error instead of the sign error rate across experiments. 

In the next section, we 
discuss the distribution of the sign error proportion (SEP)
under a hierarchical model for 
the $Y_i$'s and $\theta_i$'s, and relate this to %control of 
a marginal 
sign error rate (MSER). We then propose an adaptive nonparametric procedure that 
controls the MSER below a desired threshold regardless of the 
distribution of the $\theta_i$'s. This procedure is more powerful than BY procedure in terms of the number of rejections made, and therefore in terms of the 
number of signs inferred. The power can 
be further improved if one is willing to assume a parametric model 
for the distribution of the $\theta_i$'s. We show that 
a model-based approach to MSER control can achieve an 
optimal power asymptotically, if a model for the $\theta_i$'s 
is chosen correctly. In Section 3, we numerically compare  the nonparametric procedure
and parametric procedures to the BY procedure and an oracle MSER control procedure 
in a simulation study. 
In Section 4, we discuss an adaptive procedure 
for the somewhat different task of sign inference subject 
to fixed experimentwise type I error rate. We show how the acceptance region
of a level-$\alpha$ test of each $H_i$ may be adaptively chosen 
to minimize the MSER or maximize the power, that is, the number of sign discoveries. A discussion follows in Section 5.

\section{Sign Error Rate Control Procedures}

\subsection{Marginal Sign Error Rate}
We are interested in inferring the sign of each $\theta_i$ in the normal means model in (\ref{m1}). We test $H_i: \theta_i = 0$ using the usual level-$\alpha$ $z$-test,  and estimate $sign(\theta_i)$ by $sign(Y_i)$ if the test rejects and do not estimate the sign otherwise. We use the pair $(R_i, S_i)$ to denote the outcome of this procedure, where $R_i=1$ if $H_i$ is rejected, and $R_i = 0$ otherwise. We use $S_i$ to denote the sign estimate, with possible values 1 (positive), -1 (negative), and 0 (sign not estimated). Note that $S_i = 0$ if $R_i = 0$. A sign error is made if $S_i \cdot sign(\theta_i) = -1$. 
Let $E_i$ be the binary indicator of a sign error, so that $E_i = R_i(1-S_i \cdot sign(\theta_i))/2$. The results across experiments are summarized with $(R, E)$, where $R = \sum_{i=1}^{m} R_i$ is the total number of rejections and $E = \sum_{i=1}^{m} E_i$ is the total number of sign errors among the $m$ experiments. 
In what follows, we assume that none of the $\theta_i$'s are truly equal to 
zero. The properties of our procedures in cases where there are 
some true nulls are discussed in Section 5. 

Define the sign error proportion as $\text{SEP} = E/(R \lor 1)$. Ideally, we want to keep SEP under a desired threshold. Given a data vector $\boldsymbol{Y}$ and a experimentwise significance threshold,  the number of rejections $R$ is known but the number of sign errors $E$ is unknown since each $E_i$ depends on the unknown true parameter $\theta_i$. Therefore, SEP is an unobserved quantity that depends on the data and the unobserved parameter values. 
%Although the true parameter values are unknown, information about their distribution is available. 
However, suppose the empirical distribution of $\theta_1,\ldots, \theta_m$ is well-represented by some distribution $G$, absolutely continuous with respect to 
Lebesgue measure (and so $G(\{ 0\} ) =0$). We then assume the following model:
\begin{equation}
\theta_1,\ldots, \theta_m \;  \sim \; \text{i.i.d.}\;  G.
\label{m2}
\end{equation} 
Now (\ref{m1}) and (\ref{m2}) specify a hierarchical model. Under this hierarchical model, the probability of making a sign error for any one experiment, conditional on rejection, can be written as
\begin{equation}
\text{MSER} = \Pr( E_1=1|R_1=1 ) =\frac{\Pr(E_1=1,R_1=1)}{\Pr(R_1=1)}.
\label{MSER}
\end{equation}
%Note that we only consider the case where the experimentwise error rate $\alpha \neq 0$, thus $\Pr(R_1=1)$ would always be none zero. 
We call the quantity in (\ref{MSER}) the \textit{marginal sign error rate} (MSER). This quantity does not depend on $\boldsymbol{Y}$ or $\boldsymbol{\theta}$, just on $G$ and $\alpha$. It also determines the marginal 
distribution of the SEP:
\begin{lemma}
	Under the hierarchical  model (\ref{m1}) and (\ref{m2}),   
the conditional distribution of $R\cdot \text{SEP}$ given $R=r$ is 
$\text{binomial}(r,\text{MSER})$. 
%conditional on  $R=r$ with $r>0$, SEP has a scaled binomial distribution with mean MSER, $r\cdot SEP \;| R=r \sim Bi(r,MSER)$.
	\label{lemma:bi}
\end{lemma}
From this lemma it follows that $\Exp{\text{SEP}}  = \text{MSER} \cdot \Pr(R>0)  < \text{MSER}$. Thus by controlling MSER to be below a threshold, we bound the expected SEP under this  threshold as well. Moreover, by the following Proposition, in scenarios where $m$ is large, controlling MSER gives an accurate control over SEP.
\begin{proposition}
	 SEP converges to MSER in probability as $m \to \infty$.
	\label{coro:septomser}
\end{proposition}

In the following subsections, we propose two methods to control the MSER under a prespecified level $\alpha_S$. The first method is called the \textit{loose control} procedure, which conservatively controls MSER without parametric assumptions.  The second method is called tight control, which estimates the distribution of the $\theta_i$'s and adaptively chooses an experimentwise type I error rate $\alpha$ to maximize the number of signs estimated while 
controlling  MSER approximately below  level $\alpha_S$.

\subsection{Loose Control Procedure}
In this subsection, we develop a procedure that conservatively controls MSER. It has a good performance in ``spike and slab" scenarios where the sizes of most
of the  $\theta_i$'s are negligible compared to the measurement error, with only a few $\theta_i$'s having large values. However, for other distributions of the $\theta_i$'s it can have an MSER substantially 
below the nominal level, and so 
 we call it the loose control procedure.

The intuition for the loose control procedure is as follows: MSER can be seen as the expected number of sign errors divided by the expected number of signs inferred. With an type I error rate of $\alpha$, in the extreme case where all the $\theta_i$'s are very close to zero, we expect to infer around $\alpha \cdot m$ signs, and expect half of them to be sign errors. Hence the expected number of sign errors will be approximately $\alpha\cdot m/2$. On the other hand, the number of signs we infer is $R$. Thus intuitively we want $(\alpha m/2)/R$ to be smaller than $\alpha_S$, which suggests the following procedure: 
\begin{enumerate}
\item Find the largest $\alpha_{l}$ such that $\alpha_{l} \leq 2\alpha_SR(\alpha_{l})/m.$  
%\begin{equation}
%\alpha_{l} \leq 2\alpha_SR(\alpha_{l})/m.
%\label{eq:lc}
%\end{equation} 
\item Infer the sign for $i$th experiment if $|Y_i|> z_{1-\alpha_{l}/2}$. 
\end{enumerate} 
Here, $R(\alpha_l)$ is the number of rejections made if the 
rejection threshold is $z_{1-\alpha_l/2}$. 
We call this procedure the loose control procedure (LC). It controls MSER asymptotically in $m$: 
%We call this procedure the \textit{loose control} (LC) procedure. 
\begin{proposition} 
For the hierarchical model in  (\ref{m1}) and (\ref{m2}) and using the LC 
procedure, 
MSER $\leq  \alpha_S$ in probability as 
$m\rightarrow \infty$. 
%with $\alpha_l$ specified in (\ref{eq:lc}) controls MSER under $\alpha_S$ almost surely as $m \to \infty$. 
	\label{prop:lc}
\end{proposition}
This procedure does not provide guaranteed control of MSER for finite 
$m$ because in particular the significance threshold for each experiment 
$i$ depends to some extent on $Y_i$ through $R(\alpha_l)$. 
For small $m$ we suggest using the following procedure 
that gives exact, non-asymptotic control of MSER:
\begin{enumerate}
\item For each experiment $i$, find the largest $\alpha_{l}^i$ such that  
  $\alpha_{l}^i \leq 2\alpha_S( (R^{-i}(\alpha_{l}^i)-1)\vee 0)/m$. 
%\begin{equation}
%\alpha_{l}^i \leq 2\alpha_S( (R^{-i}(\alpha_{l}^i)-1)\vee 0)/m.
%\label{eq:lcsmall}
%\end{equation} 
\item Infer the sign for the 
$i$th experiment if $|Y_i|> z_{1-\alpha^i_{l}/2}$.  
\end{enumerate}
Here, $R^{-i}(\alpha_l^i)$ is the number of rejections made 
among all experiments except experiment $i$ if the 
significance threshold is $z_{1-\alpha_l^{i}/2}$. 
%where $R^{-i}(\alpha_{l}^i)$ is the number of rejections among all the experiments except the $i$th experiment, at level $\alpha_{l}^i$. We infer the sign for $i$th experiment if $|Y_i|> z_{1-\alpha^i_{l}/2}$.  
This procedure is slightly more conservative than LC procedure since any $\alpha_l^i$ also satisfies $\alpha_l^i \leq 2\alpha_SR(\alpha_l^i)/m$. We call this procedure the non-asymptotic loose control (NLC) procedure.
\begin{proposition}
 For the hierarchical model in  (\ref{m1}) and (\ref{m2}) and using the non-asymptotic
loose control procedure, $\text{MSER}\leq \alpha_S$. 
%	The loose control procedure with $\alpha_l^i$ specified in (\ref{eq:lcsmall}) controls MSER under $\alpha_S$. 
	\label{prop:lcsmall}
\end{proposition}

%The loose control procedure above can be simplified when $m$ is large. Denote the total number of rejections at level $\alpha$ as $R^t(\alpha)$. When $m$ is very large, $((R^{-i}(\alpha)-1)\vee 0)/m$ should be very close to $R^t(\alpha)/m$. In fact, $|((R^{-i}(\alpha)-1)\vee 0)/m - R^t(\alpha)/m| < 2/m \to 0$ as $m\to \infty$. Therefore, when $m$ is large, we can use the following simple procedure instead: Let $s=1/2$, and 

These loose control procedures are closely related to the Benjamini Yekutieli (BY) \citep{Benjamini2005} procedure, which is equivalent to finding the maximal $\alpha_{by}$ such that $\alpha_{by} \leq \alpha_SR(\alpha_{by})/m$. It is easy to see that $\alpha_{by}$ is always smaller than $\alpha_{l}$. Hence the LC procedure always infers more signs than the BY procedure. The BY procedure was proposed for controlling the 
unconditional sign error rate $\text{SER}= \text{E}[\text{SEP}|\boldsymbol{\theta}]$, 
which they called the 
``pure directional FDR''. 
%which is defined as $\text{E}[E/(R \vee 1)|\boldsymbol{\theta}]$. 
In the case that there are no true nulls, 
%The following Lemma shows that 
the loose control procedure also controls SER:
\begin{proposition}  
%	Using $\alpha_{l}$ obtained from (\ref{eq:lc}) as the level of experimentwise type I error rate, the pure directional FDR is controlled under $\alpha_S$. 
If $\theta_i\neq 0$ for all $i\in \{1,\ldots, m\}$ then both the LC and NLC procedures control the SER below $\alpha_S$. 
	\label{lemma:lcfdr}
\end{proposition}

%This procedure is closely related to the Benjamini Hochberg (BH) procedure. In fact, \citet{Benjamini2005} has proposed to use the BH procedure for mixed directional FDR control, which is equivalent to the expectation of SEP given $\boldsymbol{\theta}$. We call this the BY procedure. They proved that BY procedure could control the mixed directional FDR under a prespecified level $\alpha_S$. Now we show LC procedure is more powerful in inferring signs than BY procedure. For each group, we can calculate a $p$ value. Suppose the order statistics of $p$-values are $p_{(1)}, p_{(2)},\ldots, p_{(m)}$. Then it is easy to see that the loose control procedure for large $m$ is equivalent to finding the $\alpha_{l}$ that equals to the largest $p$-value such that 
%\begin{equation}
%p_{(i)} \leq 2\alpha_S R/m.
%\end{equation}
%This is different from BY procedure by a constant of 2 on the right side, meaning that it tends to declare more signs at the same prespecified level $\alpha_S$. Loose control procedure is simple, and it provides control over the MSER without further assumptions on the distribution of $\theta$. 

\subsection{Model Based Control Procedure}
Although the loose control procedure controls MSER without assumptions on $G$, it can be  conservative in cases where $G$ does not resemble a spike and slab distribution. In this subsection, we propose a model-based MSER control procedure that can be more powerful in terms of the number of sign inferred. 

We first discuss the oracle situation where the probability density function $G$ of $\theta_i$'s is known. The acceptance region of our test of $H_i$ is $A(\alpha) = \{Y_i: \Phi^{-1}(\alpha/2)<Y_i< \Phi^{-1}(1-\alpha/2) \}$, with $\Phi$ being the standard normal cumulative density function. We can write MSER as a function of $\alpha$ as follows:
\begin{equation}
\begin{split}
\text{MSER}(\alpha) &= \frac{\Pr(E_1 = 1, R_1=1)}{\Pr(R_1=1)}   \\
& = \frac{\text{E}_G [\text{P}_{\theta_1}(E_1 = 1, R_1=1)] }{ \text{E}_G [\text{P}_{\theta_1}(R_1=1)] } \\
& =  \frac{\text{E}_G [\text{P}_{\theta_1}( S_1 = -1 , R_1 = 1, sign(\theta_1) = 1) + \text{P}_{\theta_1}( S_1 = 1 , R_1 = 1, sign(\theta_1) = -1)] }{ \text{E}_G [\text{P}_{\theta_1}(Y_1 \not \in A(\alpha) )] } \\
& =  \frac{\text{E}_G [\text{P}_{\theta_1}( Y_1 < 0 , Y_1 \not \in A(\alpha))\textbf{1}(\theta_1 > 0) +\text{P}_{\theta_1}( Y_1 > 0 , Y_1 \not \in A(\alpha))\textbf{1}(\theta_1 < 0)]  }{ \text{E}_G [\text{P}_{\theta_1}(Y_1 \not \in A(\alpha) )] }  \\
& = \frac{\text{E}_G[B_1\textbf{1}(\theta_1 > 0)+B_2\textbf{1}(\theta_1 < 0)]}{ \text{E}_G[B_1 + B_2]},
\end{split} \label{eq:MSER}
\end{equation}
where $B_1 =\Phi(\Phi^{-1}(\alpha/2)-\theta)$ and $B_2=\Phi(\Phi^{-1}(\alpha/2)+\theta)$. 
%\begin{equation}
%\begin{split}
%& B_1 =\Phi(\Phi^{-1}(\alpha/2)-\theta),\\
%&B_2=\Phi(\Phi^{-1}(\alpha/2)+\theta) .
%\end{split} \label{eq:B}
%\end{equation}
In this case, we need to find the value of $\alpha$ such that
$\text{MSER}(\alpha) = \alpha_S$.
We denote this $\alpha$ as $\alpha_o$, and call the resulting procedure the tight control oracle (TCO) procedure. This procedure maximizes the power in inferring signs while keeping MSER at $\alpha_S$. 

In practice, $G$ is unknown and must be estimated from the data. Suppose we have an estimate $\hat G$ of $G$. By replacing $G$ by $\hat G$ in (\ref{eq:MSER}) we can obtain an empirical estimate $\widehat{\text{MSER}}$ for each value of $\alpha$, 
and in particular, find an $\alpha_e$ such that 
$\widehat{\text{MSER}}(\alpha_e) = \alpha_S$. 
We call the procedure using $\alpha_e$ instead of $\alpha_o$ the tight control empirical (TCE) procedure. 

The task of estimating $G$ from $\bs Y$ based on (\ref{m1}) and (\ref{m2}) is known as deconvolution. Current nonparametric deconvolution techniques are computationally expensive, and converge to the true $G$ slowly in $m$, 
yielding unstable results for small $m$. 
As an alternative to nonparametric deconvolution, we propose using simple parametric models to facilitate the application of the TCE procedure. The following proposition shows that under certain assumptions, the TCE procedure converges to 
the optimal TCO procedure when a correct 
parametric model for the $\theta_i$'s is used.

\begin{proposition}
  \label{prop:tcetotco}
	Suppose $\theta_1,\ldots, \theta_m\sim$ i.i.d.\ $G_\eta$ where $G_\eta$ is a member of a parametric family of distributions indexed by a finite-dimensional parameter vector $\eta$ with density function continuous in $\eta$. 
For each $m$ let $\hat \eta$ be an estimate of $\eta$, and let 
$\widehat{\text{MSER}}(\alpha)$ be the plug-in estimate of 
$\text{MSER}(\alpha)$ calculated using $G_{\hat \eta}$. 
If $\hat \eta  \stackrel{p}\rightarrow \eta$ as $m\rightarrow \infty$, then 
$\widehat{\text{MSER}}(\alpha) \stackrel{p}\rightarrow 
\text{MSER}(\alpha)$ 
%calculated using the true $g_\eta$  
and $\alpha_e \stackrel{p}\rightarrow \alpha_o$ 
as  $m\rightarrow \infty$. 
\end{proposition}

One useful model for $G$ that we explore in the next section is the 
family of 
%Specifically, we model the distribution of the $\theta_i$'s as an 
asymmetric Laplace distributions \citep{yu2005three}, which have probability density functions of the form
\begin{equation*}
	g(\theta; \mu,\tau,q) = \frac{q(1-q)}{\tau}\exp\big( -\frac{(x-\mu)}{\tau}[q - I(x\leq\mu)]  \big),
\end{equation*}
where $\mu \in \mathbb{R}$ is the location parameter, $\tau>0$ is the scale parameter, and $0< q <1$ is the skew parameter. Figure \ref{fig:ald_dist} shows the shape of ALD distributions for $q\in \{0.1, 0.3, 0.5\}$. 

\begin{figure}[!h]
	\centering
	\includegraphics[width=0.5\linewidth]{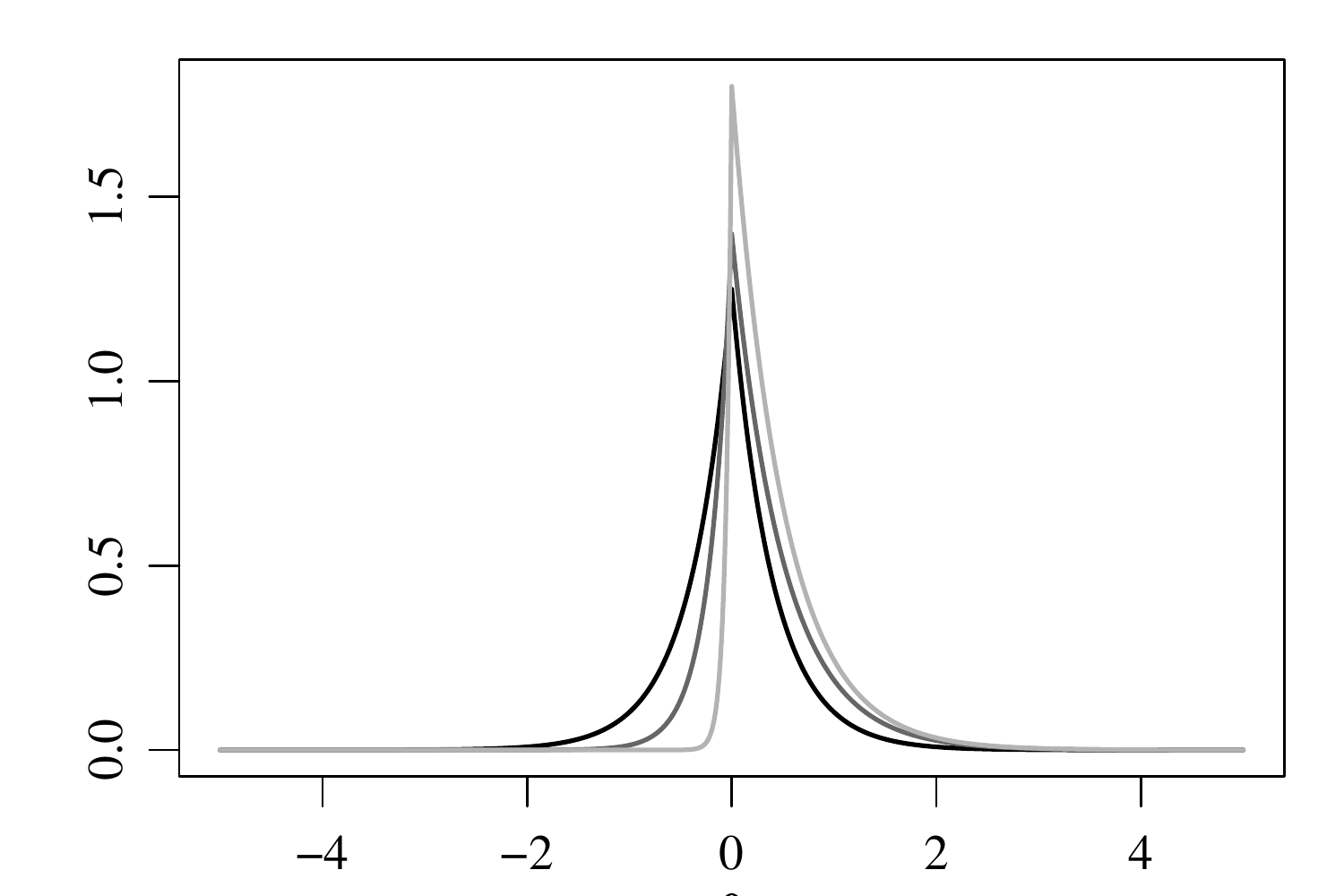}
	\caption{Shapes of asymmetric Laplace densities. The black line is the ALD density when $q = 0.5$ and $\tau = 0.2$, the darker grey line is for $q = 0.3$ and $\tau = 0.15$, and the lightest grey line is for $q=0.1$ and $\tau=0.05$. }
	\label{fig:ald_dist}
\end{figure} 
The asymmetric Laplace distribution is a flexible model for unimodal distributions with the Laplace distribution being a special case. It is more peaked at zero than a normal distribution, but also can reflect the potential skewness of 
the distribution of true effects that often exists in applications, for example, in cases where more $\theta_i$'s are positive than negative, or vice versa. 
%Its wide applicability in financial data has already been well established  \citep{kozubowski2000asymmetric}. 
For multiple testing problems were we expect that most $\theta_i$'s are
close to zero, it is natural to consider only submodels where $\mu=0$. 
%In applications, we usually assume the mode of $G$ is at 0, thus $\mu=0$. 
In this case, 
method of moment estimates for the scale and skew parameters 
may be obtained from the first and second sample moments of $\bs Y$. 
Under the hierarchical model, we have 
\begin{equation*}
	\Exp{Y} = \Exp{\Exp{Y|\theta}} = \Exp{\theta} = \frac{\tau (1-2q)}{q(1-q)},
\end{equation*}
\begin{equation*}
	\Var{Y} = \Exp{\Var{Y|\theta}} + \Var{\Exp{Y|\theta}} = 1 + \Var{\theta} = 1 + 
	\frac{\tau^2(1-2q+2q^2)}{(1-q)^2q^2}.
\end{equation*}
By setting
\begin{equation*}
	\begin{split}
		&\frac{1}{m}\sum Y_i = \frac{\tau (1-2\hat q)}{\hat q(1-\hat q)}, \\
		&  \frac{1}{m-1} \sum (Y_i - \bar y) = 1 + 
		\frac{\hat \tau^2(1-2\hat q+2\hat q^2)}{(1-\hat q)^2\hat q^2},
	\end{split}
\end{equation*}
we can solve for $\hat q$ and $\hat \tau$ to  obtain moment-based estimates of $q$ and $\tau$. %Note that although skewness is usually related to the third moment of a distribution, for this parametric model we only need the first and second moments to estimate the skew parameter $q$, making the estimation fast and stable. 

\section{Simulation Studies}
In this section we use several simulation scenarios to compare the performance of Benjamini and Yekutieli's procedure (BY), the loose control  procedure (LC), 
and a tight control empirical procedure using an asymmetric Laplace model
  for the $\theta_i$'s (TCEA). 
For each simulation scenario, 1000 datasets were simulated as follows:
First, values  $\theta_1,\ldots, \theta_m$ were independently simulated 
from a  distribution $G$. 
%distribution with probability density function $g(\theta)$. 
Then an observation vector  $\boldsymbol{Y}$ was sampled from a $N(\boldsymbol{\theta},\boldsymbol{I})$ distribution. 
%We use the three procedures to control the sign error proportion. 
For each of these datasets, the sign error proportions 
and the total numbers of signs inferred by each procedure were calculated. 
%The average actual sign error proportions and the average numbers of signs inferred using the three procedures are then compared to each other and to the oracle procedure TCO, which is optimal. 
For all procedures and simulation scenarios  the target level $\alpha_S$  was set to be 10\%.
Simulations were run for $q \in \{0.1, 0.3, 0.5\}$ and for five different values of $\tau$ for each level of $q$. The ranges of the $\tau$ values were chosen so that SEP ranged between   $10 \%$ to $30\%$ when the  experimentwise type I error rate $\alpha = 0.05$.

The results 
for several simulation scenarios with $m=5000$ 
are summarized in Figure \ref{fig:ald}. Overall, the TCEA procedure performs nearly as well as the TCO procedure. Both procedures control SEP at the prespecified level $\alpha_S = 0.1$, and infer many more signs than the BY and LC procedures, with BY being the least powerful of the three. The difference between TCEA and LC or BY becomes larger as $\tau$ increases. 

\begin{figure}[!h]
	\centering
	\includegraphics[width=0.95\linewidth]{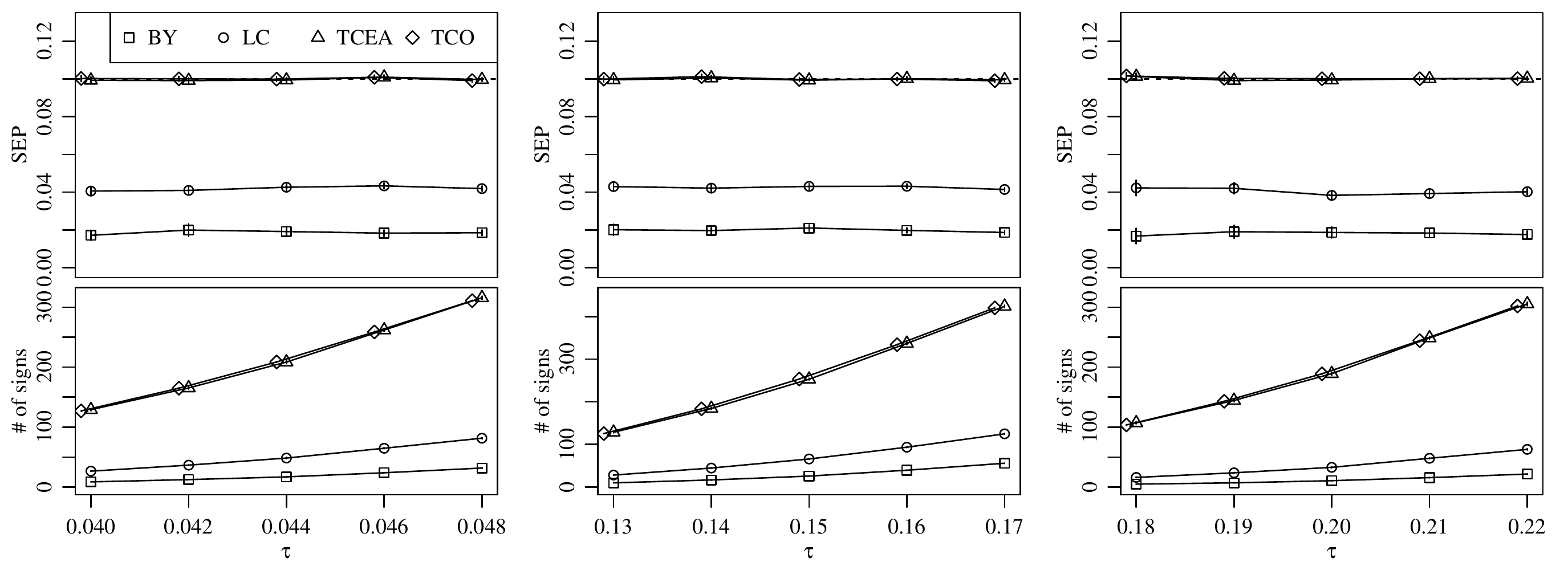}
	\caption{Comparison of the three procedures when $m=5000$ and the $\theta_i$'s have an asymmetric Laplace  distribution.  The skewness parameter $q$ is set to be 0.1 in the left column, 0.3 in the middle column, and 0.5 in the right column. Vertical bars around each plotting character correspond to $\pm 1.96$ Monte Carlo standard errors. }
	\label{fig:ald}
\end{figure}

When number of experiments is large, the TCEA procedure is very close to the TCO procedure as our asymptotic result predicts. However, when $m = 100$, TCEA and TCO show some differences. The results for several simulations with $m=100$ are summarized in Figure \ref{fig:aldss}. 
In this situation, %we see that TCEA is less stable than TCO in controlling the SEP. While 
TCEA still performs better than BY or LC in terms of the power to infer signs. 
Also, we see that for some cases, the SEP of the oracle procedure does not attain the nominal level of 0.1. This is because tight control procedure is designed to keep MSER under the nominal level $\alpha_s$. As illustrated before, controlling MSER under $\alpha_S$ gives an accurate control over the expected SEP when $m$ is large. When $m$ is small,  the probability of making no rejections across all experiments is non-negligible, and MSER is slightly larger than expectation of SEP. In this case, instead of keeping the average SEP at $\alpha_S$, TCO keeps it under $\alpha_S$, making the result slightly conservative. 

\begin{figure}[!h]
	\centering
	\includegraphics[width=0.95\linewidth]{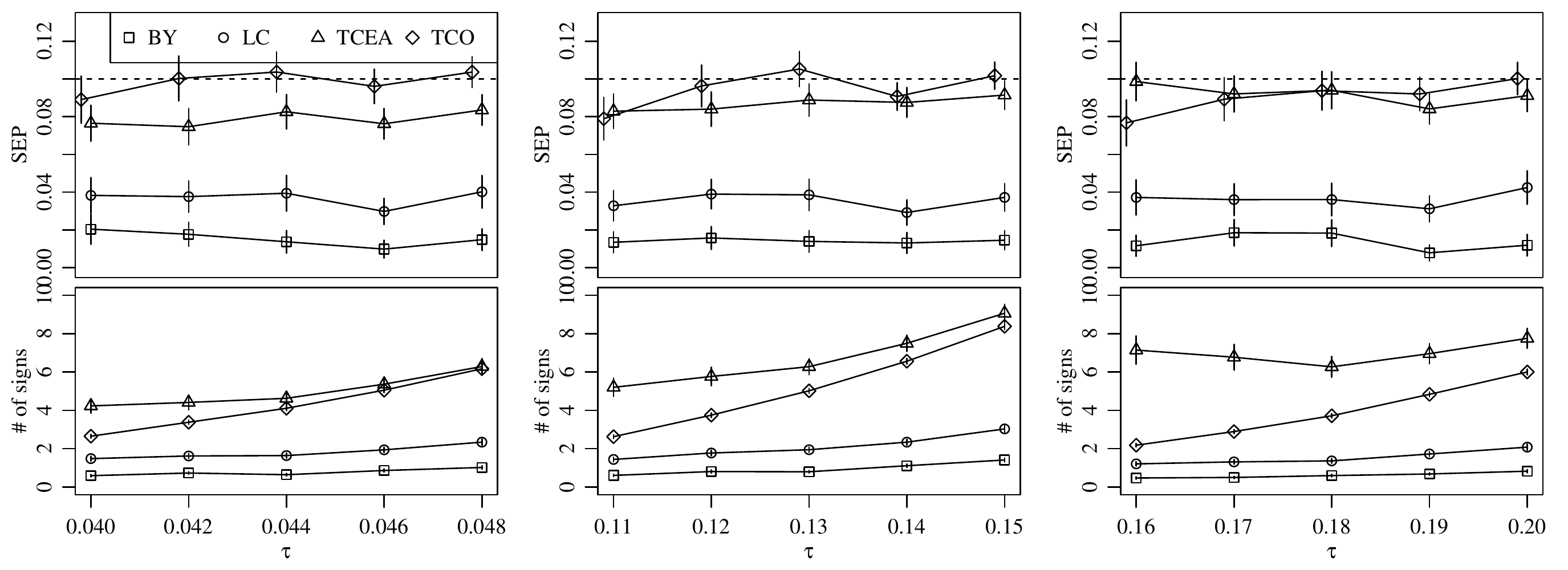}
	\caption{Comparison of the four procedures under the same settings as in Figure \ref{fig:ald} but $m = 100$. }
	\label{fig:aldss}
\end{figure}

Finally, we study the situation when $G$ is a spike and slab distribution. The spike is a unimodal distribution with mean zero and small variance, and the slab is a uniform distribution.  For two asymmetric cases 
($q\in\{0.1,0.3\}$) 
the slab is the uniform distribution on (2,4). For the symmetric case ($q = 0.5$), the slab is the uniform  distribution on  $(-4, -2) \cup (2, 4)$.  
In each case, 
the proportion of $\theta_i$'s that are sampled from the slab is $1 \%$. Comparisons of the three procedures and TCO  are summarized in Figure \ref{fig:ss}. As expected, the LC procedure overall has better performance than the BY and TCEA procedures. As the variance of the spike grows larger, the differences  between the $\theta$-values sampled from the spike and the $\theta$-values sampled from the slab becomes smaller, and the multimodal spike and slab distribution becomes closer and closer to a unimodal distribution that can be well-represented by a member of the asymmetric Laplace family. In such scenarios, TCEA does well in terms of maintaining MSER and inferring signs.  

%Both LC and BY procedures are affected less by the number of experiments $m$. One interesting result is that TCEA infers many more signs than the other procedures. According to our simulation settings, there should be 5 of the $\theta_i$'s that have relative large values, and TCO and LC are doing well in inferring signs of these 5 $\theta_i$'s. While TCEA adapts to the shape of the empirical distribution of $\theta_i$'s, and infers signs for some $\theta_i$'s that are sampled from the ``spike".  

\begin{figure}[!h]
	\centering
	\includegraphics[width=0.95\linewidth]{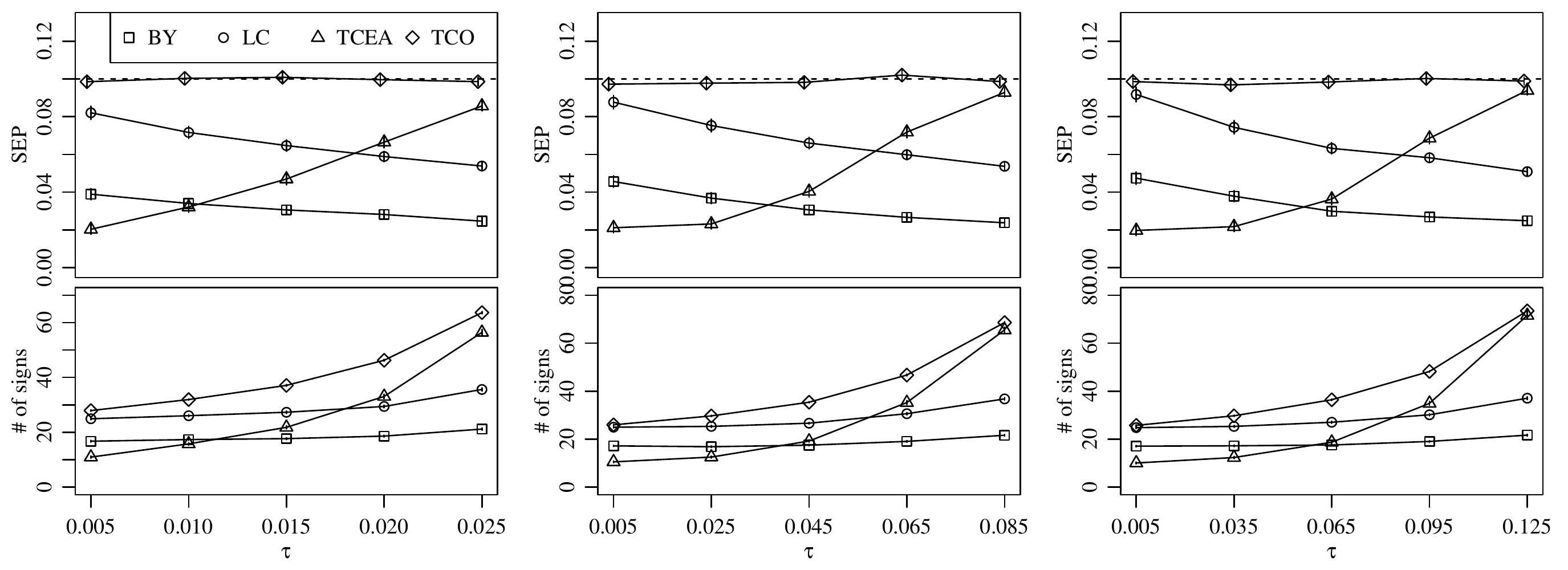}
	\caption{Comparisons of the three procedures when $m=5000$ and under a spike and slab distribution for the $\theta_i$'s. From left to right, the spike is sampled from an asymmetric Laplace distribution with $q = 0.1, 0.3, 0.5$, respectively. %1\% of the $\theta_i$'s comes from the slab. For asymmetric cases, the slab is sampled from Unif(2,4). For the symmetric case ($q = 0.5$), the slab is uniformly sampled from $(-4, -2) \cup (2, 4)$. }
}
	\label{fig:ss}
\end{figure}

%\begin{figure}[!h]
%	\centering
%	\includegraphics[width=0.95\linewidth]{plot_ss_s}
%	\caption{Comparison of the four procedures under the settings with $g$ being ``spike and slab'' distribution. Same setting as Figure \ref{fig:ss} with smaller $m = 500$. }
%	\label{fig:sss}
%\end{figure}

\section{MSER and MSDR Optimization Subject to Type I Error Control}
We have discussed controlling MSER under a prespecified level by choosing an appropriate significance threshold. In this section, we study the relationship between MSER and the shape of the acceptance region when the level $\alpha$ for the experimentwise type I error rate is held fixed.  We show how to minimize the MSER while maintaining the experimentwise type I error rate. \citet{storey2007optimal} has proposed a general framework for maximizing the statistical power of a test while maintaining the experimentwise type I error rate. \citet{wasserman2006weighted} and \citet{Dobriban2015} studied a weighted Bonferroni method to control family-wise type I error rate while maximizing the power.  As illustrated in \citet{gelman2014beyond} and  \citet{owen2016confidence}, a high sign error rate occurs when the error variance is large compared to the true effect size. We show that other than the error variance, the shape of the acceptance region is another crucial factor in determining the sign error rate.

In addition to MSER, we define the Marginal Sign Discovery Rate (MSDR) as $
	\text{MSDR} = \Pr(R_1 = 1)$.
This quantity measures the expected proportion of the number of experiments with a sign inferred among all of the experiments since 
\begin{equation*}
	\text{MSDR} = \Pr(R_1 = 1) =   \frac{\sum_{i=1}^{m} \Pr(R_i=1)}{m} = \frac{\sum_{i=1}^{m} \Exp{\textbf{1}(R_i=1)}}{m} =  \frac{\Exp{\sum_{i=1}^{m} \textbf{1}(R_i=1)}}{m} = \text{E}\left[ \frac{R}{m} \right].
\end{equation*}

Both MSER and MSDR are affected by the acceptance region of the test. The usual acceptance region for each $H_i$ is $A = (\Phi^{-1}(\alpha/2),  \Phi^{-1}(1-\alpha/2))$, which corresponds to the uniformly most accurate unbiased (UMAU) test. Following the ideas of \citet{yu2016adaptive}, we can construct a class of acceptance regions that corresponds to all level $\alpha$ two-sided tests $	A(\alpha,s) = \{Y_i:  \Phi^{-1}(\alpha s)<Y_i< \Phi^{-1}(1-\alpha(1-s)) \}$,
where $s \in (0,1)$ is a constant. Thus even if the level $\alpha$ is fixed, we can change the acceptance region by varying its endpoints. When $s<1/2$, the acceptance region tends to cover more negative observations and less positive observations. When $s>1/2$, the acceptance region tends to cover more positive observations and less negative observations. As $s \to 0$ or $1$, the two-sided test converges to a one-sided test with an acceptance region of either $( \Phi^{-1}(\alpha), \infty)$ or $(-\infty,  \Phi^{-1}(1-\alpha)) $. We now examine which $s$ value minimizes MSER and which $s$ value maximizes MSDR when the experimentwise type I error rate $\alpha$ is held fixed.  Similar to (\ref{eq:MSER}), we can express the MSER as 
\begin{equation*}
	\begin{split}
		\text{MSER}(A(\alpha,s)) &= \frac{E[B_1\textbf{1}(\theta_1 > 0)+B_2\textbf{1}(\theta_1 < 0)]}{ E[B_1 + B_2]}, \\
		\text{MSDR}(A(\alpha,s)) & = E[B_1 + B_2],
	\end{split}
\end{equation*}
where $B_1 =\Phi(\Phi^{-1}(\alpha s)-\theta)$ and $B_2=\Phi(\Phi^{-1}(\alpha(1-s))+\theta)$.

If we fix $\alpha$, MSER and MSDR can be seen as function of $s$. Under our models, we turn the minimization of MSER and maximization of MSDR into two one-parameter optimization problems: Denote
\begin{equation*}
	\begin{split}
		s^D &= \text{arg} \max_s \text{MSDR}(s) \\
		s^E &= \text{arg} \min_s \text{MSER}(s).
	\end{split}
\end{equation*}

Interestingly the UMAU procedure, where $s=  s^U = 1/2$, does not always maximize the expected power, and the $s$ that maximizes the MSDR does not necessarily minimizes the MSER, vice-versa. We use a simple numerical example to illustrate this. Suppose $\theta_i$'s are sampled from a shifted chi-square distribution $\chi^2_3 - 3$. By numerical evaluation, the results are summarized in Table \ref{ex:dif}.
	\begin{table}[!h] \centering
		\begin{tabular}{|c|c|c|c|c|}
			\hline  & $s$ value & $A(s,0.05)$ & MSER($\%$) & MSDR \\ 
			\hline $s^U$ & 0.5 & (-3.92, 3.92) & 3.01 & 0.189 \\ 
			\hline $s^D$ & 0.683 & (-3.65, 4.30) & 2.79 & 0.193 \\ 
			\hline $s^E$ & 0.829 & (-3.45, 4.80) & 2.71 & 0.190  \\ 
			\hline 
		\end{tabular} 
		\caption{Comparison of the usual acceptance region, the acceptance region that maximizes MSDR, and the acceptance region that minimizes MSER}
		\label{ex:dif}
	\end{table}

On the other hand, \citet{storey2007optimal} noticed that when $\theta \sim N(0, \sigma_b^2)$, the test that maximizes expected power is the UMAU test. Here we prove a more general theorem that the UMAU test actually both maximizes expected power and minimizes MSER when the distribution of $\theta$ is symmetric.
\begin{proposition}
	If $G$ is a distribution that is symmetric with respect to 0, the two-sided test that maximizes MSDR and minimizes the MSER is the UMAU test, i.e. $s^D = s^E = 1/2$. 
	\label{theorem:mser}
\end{proposition}
Thus in applications where $\alpha$ is held fixed, if we believe that the distribution of the $\theta_i$'s is symmetric, we should use the usual acceptance region. In situations where we suspect this distribution to be asymmetric, then using either $S^D$ or $S^E$ can lead to a test with either higher MSDR or lower MSER. However, identifying $S^D$ or $S^E$ requires $G$ to be known. Similar to the TCE procedure, in practice we replace $G$ with an estimate $\hat G$ and obtain empirical estimates $\widehat{\text{MSDR}}$  and $\widehat{ \text{MSER}}$, and then obtain $S^D$ or $S^E$ by maximizing $\widehat{ \text{MSDR}}$ or minimizing $\widehat{ \text{MSER}}$.

\section{Discussion}
In this article, we use the MSER as a measure of sign errors in multiple testing settings. We proposed two types of procedures to control MSER, loose control procedure and tight control procedure. Loose control procedure can be  conservative but is robust to the distribution of the 
$\theta_i$'s, while the tight control  procedure is more  powerful but assumes the distribution of $\theta_i$'s is a member of a known parametric model.

The loose control procedure proposed in this paper is closely related to the BY procedure. Unlike the derivation for the BY procedure, we derive the LC procedure from the perspective of controlling the MSER, which is a quantity measuring the probability of making a sign error under a hierarchical model. We assume that there are no ``true nulls'' in this paper, because in many applications true nulls do not exist. By assuming no true nulls, the loose control procedure we derived is more powerful than the BY procedure in terms of the number of inferred signs. If it is believed that the true nulls do exist, the loose control procedure can still control the SER, although control over MSER depends on how we define a sign error when $\theta_i = 0$. If we define that when $\theta_i=0$, either claiming $\theta$ is positive or negative is correct, the loose control procedure stays the same as proposed in this paper. If we define that when $\theta_i=0$, either claiming $\theta$ is positive or negative is wrong, then the BY procedure should be used since it also controls the mixed directional FDR, where any sign declaration of $\theta_i =0$ is considered as a sign error.

%should be modified according to the definition of a sign error when $\theta=0$. If we define that when $\theta=0$, either claiming $\theta$ is positive or negative is correct, LC procedure stays the same as proposed earlier in this paper. If we define that when $\theta=0$, either claiming $\theta$ is positive or negative is wrong, LC procedure will be modified to be of similar form as BY procedure.  However, one should note that tight control procedure is not affected by whether there are ``true nulls'' or not. Though different model or deconvolution technique should be used for the estimation of $g$ function in the case where ``true nulls'' are believed to exist.

We also discussed varying the endpoints of the acceptance region to reduce MSER and increase MSDR when the type I error rate is fixed. This can be combined with the tight control procedure, leading to a new procedure: Choose $\alpha$ and $s$ such that 
\begin{equation*}
	\begin{split}
		& (\alpha,s) = \arg \max_{(\alpha, s)} \widehat{MSDR}(A(\alpha,s)) \\
		& \text{such that\quad} \widehat{\text{MSER}}(A(\alpha,s)) < \alpha_S.
	\end{split}
\end{equation*}
Given an estimate $\hat G$ of $G$, the solution for $(\alpha, s)$ can be obtained numerically. This procedure can potentially increase the power in inferring signs. However, the performance of this procedure is more unstable since the optimization task here is more complicated.

\section*{Acknowledgment} This research was partially supported by
NSF grant DMS-1505136. 

\section*{Appendix}

\begin{proof}[Proof of  Lemma \ref{lemma:bi}] 
	Note that $(Y_1, \theta_1),\ldots, (Y_m, \theta_m)$ are an i.i.d sample from the hierarchical  model (\ref{m1}) and (\ref{m2}). For $H_i$,  $\forall i \in \{1,..., m\} $, given that it is rejected, the probability of making a sign error is $\Pr(E_i = 1|R_i =1)$, which is MSER as specified in (\ref{MSER}).  Given that $R=r$ hypotheses are rejected, the total number of sign errors should follow a binomial distribution, i.e.
		$E\;| R=r \sim Bi(r,\text{MSER})$.
	Thus $
	R \cdot \text{SEP}\;| R=r \sim Bi(r,\text{MSER})$.
\end{proof}

\begin{proof}[Proof of  Proposition \ref{coro:septomser}]
We just need to show that $\text{SEP} - \text{MSER} \to 0$ in probability, which is to show $\text{SEP}  -\Exp{\text{SEP}} +\Exp{\text{SEP}} - \text{MSER} \to 0 $ in probability. Since $\Exp{\text{SEP}} = \text{MSER} \cdot \Pr(R>0) = \text{MSER} \cdot (1-\Pr(R_1=0)^m)$, we have $\Exp{\text{SEP}} \to \text{MSER}$ in probability as $m \to \infty$ (note $\Pr(R_1=0)<1$ in our setting). Now we just need to show that $\text{SEP} - \Exp{\text{SEP}} \to 0$ in probability, which can be done by showing $\Exp{(\text{SEP} - \Exp{\text{SEP}})^2} \to 0$. We have
\begin{equation*}
\begin{split}
\Exp{(\text{SEP} - \Exp{\text{SEP}})^2} & = \Var{\text{SEP}}  =  \Var{\Exp{\text{SEP}|R}} + \Exp{\Var{\text{SEP}|R}} \\
& = \Var{\text{MSER} \cdot \textbf{1}(R>0)} +  \Exp{\frac{R\cdot \text{MSER}(1-\text{MSER})}{R^2} \textbf{1}(R>0)} \\
& = \text{MSER}^2\cdot \Pr(R>0)(1- \Pr(R>0)) +  \text{MSER}(1-\text{MSER}) \cdot \Exp{\frac{1}{R}\textbf{1}(R>0)}.
\end{split}
\end{equation*}
The first part goes to 0 because $\Pr(R>0) \to 1$ as $m \to \infty$. The second part goes to 0 because $R$ follows a binomial distribution $Bi(m, \Pr(R_1=1))$, and 
\begin{equation*}
\begin{split}
\Exp{\frac{1}{R}\textbf{1}(R>0)} & < \Exp{\frac{2}{R+1}\textbf{1}(R>0)} < 2  \Exp{\frac{1}{R+1}\textbf{1}(R>0)} \\ &= 2  \Exp{\frac{1}{R+1}} - 2 \Exp{\frac{1}{0+1}\textbf{1}(R=0) }  \\ &= \frac{2}{(m+1)\Pr(R_1 = 1)}\cdot (1-(1-\Pr(R_1=1))^{m+1}) - 2\Pr(R=0) \to 0 
\end{split}
\end{equation*}
as $m \to \infty$. Therefore, $\text{SEP} - \text{MSER} \to 0$ in probability. 
\end{proof}

%\noindent   \textbf{Lemma} \ref{lemma:ineq} 

%\begin{proof}	
% The expectation of \text{SEP} is $ESEP = E[SEP]$. Given $R$, when $R>0$, $E[SEP|R] = \text{MSER}$ by \ref{lemma:bi}. When $R= 0$, $E[SEP|R] = 0$ due to definition. Therefore 
%  \begin{equation}
%  \begin{split}
%  ESEP &= E[SEP\textbf{1}(R>0) + SEP\textbf{1}(R=0)] \\
%  &= E[E[SEP\textbf{1}(R>0) + SEP\textbf{1}(R=0) |R]] \\ 
%  &= E[E[SEP\textbf{1}(R>0)|R]] +E[E[ SEP\textbf{1}(R=0) |R]] \\
%  & = E[\textbf{1}(R>0) E[SEP|R, R>0]] + E[\textbf{1}(R=0) E[ SEP |R=0]] \\
%  & = E[\textbf{1}(R>0) \text{MSER} ] + 0 \\
%  & = \Pr(R>0)\text{MSER}
%  \end{split}
%  \end{equation}

% Also, $\Pr(R>0) = 1-(\Pr(R_1=0))^m$. Notice that $\Pr(R_1=0) < 1$, unless the acceptance region is an empty set, which is not valid. Thus $\Pr(R>0) \to 1$ as $m \to \infty$, and $ESEP \to \text{MSER}$. 
% \end{proof}

Before proving Proposition \ref{prop:lc} and Proposition \ref{prop:lcsmall}, we first prove the Lemma below.
\begin{lemma}
	Let $A(\alpha,s) = \{y: \Phi^{-1}(\alpha s) < y < \Phi^{-1}(1-\alpha (1-s))$. Let 
	\begin{equation*}
		\gamma(A(\alpha,s)) = \text{E}_G[B_1(A(\alpha,s))\textbf{1}(\theta > 0)+B_2(A(\alpha,s))\textbf{1}(\theta < 0)], 
	\end{equation*}
	we have that $\gamma(A(\alpha,s)) \leq \alpha s \pi_0 +\alpha (1-s) (1-\pi_0)$,
	where $\pi_0 = \Pr(\theta > 0)$.
	\label{lemma:lc}
\end{lemma}
\begin{proof} Under the hierarchical model we have, 
	\begin{equation*}
		\gamma(A(\alpha,s)) = \text{E}_G[\Phi(\Phi^{-1}(\alpha s) - \theta)\textbf{1}(\theta > 0) + \Phi(\Phi^{-1}(\alpha (1-s)) + \theta)\textbf{1}(\theta < 0)]. 
	\end{equation*}
	Denote $\gamma(A(\alpha,s)) = \gamma_1 + \gamma_2$ where $
		\gamma_1 =  \text{E}_G[\Phi(\Phi^{-1}(\alpha s) - \theta)\textbf{1}(\theta > 0)]$, and $
		\gamma_2 =  \text{E}_G[\Phi(\Phi^{-1}(\alpha (1-s)) + \theta)\textbf{1}(\theta < 0)]$. 
	Suppose the probability density function of $G$ is $g$,for $\gamma_1$ we have
	\begin{equation}
		\begin{split}
			\gamma_1 = & \text{E}_G[\Phi(\Phi^{-1}(\alpha s) - \theta)\textbf{1}(\theta \geq 0)] \\
			& = \int_{0}^{\infty} \Phi(-\theta + \Phi^{-1}(\alpha s)) g(\theta) d\theta \\
			& = \Phi(-\theta + \Phi^{-1}(\alpha s))G(\theta) |^{\infty}_{0} + \int_{0}^{\infty} \phi(-\theta + \Phi^{-1}(\alpha s)) G(\theta)d\theta \\
			& = -\alpha s (1-\pi_0)  + \int_{0}^{\infty} \phi(-\theta + \Phi^{-1}(\alpha s)) G(\theta)d\theta \\
			& \leq -\alpha s (1-\pi_0) + \int_{0}^{\infty} \phi(-\theta + \Phi^{-1}(\alpha s)) d\theta \\
			& = -\alpha s (1-\pi_0)  + \alpha s = \alpha s \pi_0
		\end{split}
	\end{equation}
	For $\gamma_2$ we have
	\begin{equation}
		\begin{split}
			\gamma_2 = & \text{E}_G[\Phi(\Phi^{-1}(\alpha (1-s)) + \theta)\textbf{1}(\theta \leq 0)] \\
			& = \int_{-\infty}^0 \Phi(\theta + \Phi^{-1}(\alpha 
			(1- s)) )g(\theta) d\theta \\
			& = \Phi(\theta + \Phi^{-1}(\alpha (1- s)))G(\theta) |_{-\infty}^{0} -\int^{0}_{-\infty} \phi(\theta + \Phi^{-1}(\alpha (1- s))) G(\theta)d\theta \\
			& = \alpha (1- s)(1-\pi_0)  -\int^{0}_{-\infty} \phi(\theta + \Phi^{-1}(\alpha (1- s))) G(\theta)d\theta \\
			&\leq \alpha (1- s)(1-\pi_0) 
		\end{split}
	\end{equation}
	Therefore $E(A(\alpha,s)) = \gamma_1 + \gamma_2 \leq \alpha s \pi_0 +\alpha (1-s) (1-\pi_0)$.
\end{proof}

%\noindent \textbf{Proposition} \ref{prop:lc}
%\begin{proof}
%	By Lemma \ref{lemma:lc} we know that for usual acceptance region ($s= 1/2$),
%	\begin{equation}
%	\gamma \leq \pi_0 \alpha/2  +\alpha (1-\pi_0)/2 = \alpha/2.
%	\label{ine:me}
%	\end{equation} 
%	Also note that $\text{MSER} = \gamma/\beta$, where $\beta = \Pr(R_1 = 1)$. Thus $
%	\gamma \leq \alpha_{l}/2 \leq \alpha_S R/m$, and we have $\gamma / (R/m) \leq \alpha_S$. Asymptotically, we have $R/m \to \Pr(R_1 = 1)=\beta$ in probability by Law of Large Number, thus we have $\gamma / \beta \leq \alpha_S$ using loose control procedure.
%\end{proof}

\begin{proof}[Proof of Proposition \ref{prop:lcsmall}]
 Denote $R^t$ as the total number of rejections. We have 
	\begin{equation*}
		\Exp{R^t/m} = \Exp{\sum \textbf{1}(R_j=1)}/m =\sum \Pr(R_j = 1) /m = \Pr(R_i =1),
	\end{equation*}
	where the last step is because of the exchangeability of the model. Again, we write $\text{MSER} = \gamma/\beta$, where $\gamma = \Pr(E_i = 1, R_i = 1)$ and $\beta =  \Pr(R_i = 1)$. Since $\alpha_{l}^i$ is independent of $Y_i$, and by Lemma \ref{lemma:lc} and letting $s=1/2$, we have
	\begin{equation*}
		\Pr(E_i = 1, R_i = 1 | \alpha_{l}^i) \leq  \alpha_{l}^i/2 \leq \alpha_S((R(\alpha_{l}^i)-1)\vee 0)/m \leq \alpha_S R^t/m.
	\end{equation*}
	Thus $\gamma = \Exp{\Pr(E_i = 1, R_i = 1 | \alpha_{l}^i)} \leq \alpha_S \Exp{R^t/m} =\alpha_S\beta$. Therefore $\text{MSER} \leq \alpha_S$.
\end{proof}

\begin{proof}[Proof of Proposition \ref{lemma:lcfdr}]
This Proposition follows from \citet{Benjamini2005} Theorem 1 and Corollary 3. To modify the proof for LC procedure, we should replace the $kq/m$ in equation (4) in \citet{Benjamini2005} with $2kq/m$. Then it is easy to see that the SER can be controlled under $q$, which is the $\alpha_S$ we have in this paper. Since LC is more conservative than LC, NLC also controls SER below $\alpha_S$.
\end{proof}
 
\begin{proof}[Proof of  Proposition \ref{prop:lc}]
	This is implied by  Proposition \ref{lemma:lcfdr}. Note that when $m \to \infty$, $\text{SER} - \text{MSER} \to 0$ in probability according to the proof of Proposition \ref{coro:septomser}, hence MSER $<\alpha_s$ in probability. 
\end{proof}

\begin{proof}[Proof of Proposition  \ref{prop:tcetotco}]
	We first show that $\widehat{\text{MSER}} \stackrel{p} \rightarrow \text{MSER}$. Since both $B_1$ and $B_2$ are integrable and the probability density function $g_{\eta}$ of $G$ is continuous in $\eta$, $\text{E}_G[B_1+B_2] = \int (B_1+B_2) g_{\eta}(\theta) d \theta$ is a continuous function of $\eta$ and it is always nonzero. Similarly, $\text{E}_G[B_1\textbf{1}(\theta \geq 0)+B_2\textbf{1}(\theta \leq 0)]$ is a continuous function in $\eta$. Therefore, MSER is a continuous function in $\eta$. Note that the difference between MSER and $\widehat{\text{MSER}}$ is that the former uses $\eta$ and the later uses $\hat \eta$. If $\hat \eta  \stackrel{p}\rightarrow \eta$, then we have $\widehat{\text{MSER}} \to \text{MSER}$ by Continuous Mapping Theorem. 
	
	Since $\alpha_o$ is the unique solution such that $\text{MSER}(A(\alpha_o)) - \alpha_S=0$, and $\alpha^e$ is the unique solution such that $\widehat{\text{MSER}}(A(\alpha_e)) - \alpha_e = 0$, we have $\alpha_e  \stackrel{p}\rightarrow \alpha_o$ by M-estimator theory (Lemma 5.10, \citet{van1998asymptotic}).

\end{proof}

\begin{proof}[Proof of  Proposition \ref{theorem:mser}]
   We first show that $s=1/2$ maximizes the MSDR. The MSDR can be written as 
	\begin{equation*}
		\begin{split}
			\text{MSDR}(s) &= \int_{-\infty}^{\infty} (B_1(\theta,s) + B_2(\theta,s) ) g(\theta)d \theta \\
			& = \int_{-\infty}^{0} (B_1(\theta,s) + B_2(\theta,s) ) g(\theta)d \theta +\int_{0}^{\infty} (B_1(\theta,s) + B_2(\theta,s) ) g(\theta)d \theta 
		\end{split}
	\end{equation*}
	Since $g$ is symmetric,
	\begin{equation*}
		\begin{split}
			\text{MSDR}(s)
			& = \int_{0}^{\infty} (B_1(-\theta,s) + B_2(-\theta,s) ) g(\theta)d \theta +\int_{0}^{\infty} (B_1(\theta,s) + B_2(\theta,s) ) g(\theta)d \theta \\
			& = \int_{0}^{\infty} ((B_1(\theta,s) + B_2(\theta,s) +B_1(-\theta,s) + B_2(-\theta,s) )g(\theta) d\theta
		\end{split}
	\end{equation*}
	Now we prove that the integrand is maximized when $s=1/2$, which does not depend on $\theta$. Thus $\text{MSDR}(s)$ is maximized when $s=1/2$. The integrand can be written as $H(s)g(\theta)$ where
	% 	We prove it for any symmetric discrete distribution. Without loss of generality, we suppose the null is $\theta = 0$. Suppose $g$ is a distribution with support $\{...,-\theta_i,...,-\theta_1,0,\theta_1,...,\theta_i,...\}$ with $\theta_i$s all positive, and corresponding probability mass $(p_q,...,p_1,p_0,p_1,...,p_q)$. Then the expected power 
	% 	\begin{equation}
	% 	\text{MSDR}(s) = \sum p_i(B_1(\theta_i,s) + B_2(\theta_i,s) + B_1(-\theta_i,s) + B_2(-\theta_i,s) ) 
	% 	\end{equation} 	
	\begin{equation}
		H(s) = \Phi(\Phi^{-1}(\alpha s) - \theta)+\Phi(\Phi^{-1}(\alpha (1-s)) + \theta)+ \Phi(\Phi^{-1}(\alpha s) + \theta)+ \Phi(\Phi^{-1}(\alpha (1-s)) - \theta)
		\label{eq:integrand}
	\end{equation}
	Taking the derivative with respect to $s$, we have
	\begin{equation}
		\begin{split}
			H(s)^\prime &= \frac{\phi(\Phi^{-1}(\alpha s)-\theta)}{\phi(\Phi^{-1}(\alpha s))} +\frac{\phi(\Phi^{-1}(\alpha s)+\theta)}{\phi(\Phi^{-1}(\alpha s))} -\frac{\phi(\Phi^{-1}(\alpha (1-s))-\theta)}{\phi(\Phi^{-1}(\alpha (1-s)))} - \frac{\phi(\Phi^{-1}(\alpha (1-s))+\theta)}{\phi(\Phi^{-1}(\alpha (1-s)))} \\
			& = c_1(exp(\Phi^{-1}(\alpha s)\theta)+exp(-\Phi^{-1}(\alpha s)\theta) 
			-exp(\Phi^{-1}(\alpha (1-s))\theta) -exp(-\Phi^{-1}(\alpha (1-s))\theta)), 
			\label{eq:solution}
		\end{split}
	\end{equation}
	where $c_1$ is a positive constant. It's easy to see that $s=1/2$ is one solution to $H(s)^\prime = 0$. Now we show that $H(s)$ is actually concave, hence $s=1/2$ maximizes $H(s)$ for every $\theta > 0$. Therefore $s=1/2$ maximizes $\text{MSDR}(s)$. By taking derivative of $H(s)^\prime$ with respect to $s$ and rearrange, we obtain
	\begin{equation*}
		\begin{split}
			H(s)^{\prime\prime} &= c_2(exp( (\Phi^{-1}(\alpha s)+\theta)^2/2 )  
			+exp( (\Phi^{-1}(\alpha (1-s))+\theta)^2/2 )
			-exp( (\Phi^{-1}(\alpha s)-\theta)^2/2 )
			\\ & \quad \quad -exp( (\Phi^{-1}(\alpha (1-s))-\theta)^2/2 )
			),
		\end{split}
	\end{equation*}
	where $c_2$ is a positive constant. Since $\Phi^{-1}(\alpha s) < 0 $ and $\theta > 0 $ (the integral is from 0 to $\infty$), we have 
	\begin{equation*}
		|\Phi^{-1}(\alpha s)-\theta| = |\Phi^{-1}(\alpha s)| + |\theta| \geq |\Phi^{-1}(\alpha s) + \theta|.
	\end{equation*}
	Thus 
	\begin{equation*}
		exp( (\Phi^{-1}(\alpha s)+\theta)^2/2 ) -exp( (\Phi^{-1}(\alpha s)-\theta)^2/2 ) < 0.
	\end{equation*}
	Similarly 
	\begin{equation*}
		exp( (\Phi^{-1}(\alpha (1-s))+\theta)^2/2 ) -exp( (\Phi^{-1}(\alpha (1-s))-\theta)^2/2 ) < 0.
	\end{equation*}
	Therefore $H(s)^{\prime\prime} < 0$, and $s=1/2$ maximizes $\text{MSDR}(s)$.
	
	To show MSER is minimized by $1/2$, we can first show that $s=1/2$ minimizes $\gamma$, using the same technique as previous part of this proof. Then by noticing that $\text{MSER} = \gamma/\text{MSDR}$, we know $s=1/2$ minimizes MSER.
\end{proof}

\bibliographystyle{chicago}
\bibliography{refs}

\end{document}